\documentclass[11pt]{article}
\usepackage{fullpage}
\usepackage{epsfig}
\usepackage{amsmath,amssymb,amsthm}
\usepackage{color}
\usepackage{authblk}

\newtheorem {thm}{Theorem}[section]
\newtheorem {lem}[thm]{Lemma}

\newcommand{\CH}{\mbox{$\cal{H}$}}

\newcommand{\SHF}{\mbox{$\mathsf{SHF}$}}
\newcommand{\FPC}{\mbox{$\mathsf{FPC}$}}

\newcommand{\T}{\mbox{${\mathsf T}$}} 
\newcommand{\Z}{\mbox{${\mathsf Z}$}}

\newcommand{\A}{\mbox{$\mathsf{A}$}}
\newcommand{\B}{\mbox{$\mathsf{B}$}}
\newcommand{\C}{\mbox{$\mathsf{C}$}}

\begin{document}

\title{On tight bounds for binary frameproof codes}
\author[1]{Chuan Guo}
\author[2]{Douglas~R.~Stinson\thanks{D.~Stinson's research is supported by NSERC discovery grant 203114-11.}}
\author[3]{Tran van Trung}
\affil[1,2]{David R.\ Cheriton School of Computer Science, University of Waterloo,
Waterloo, Ontario, N2L 3G1, Canada}
\affil[3]{Institut f\"ur  Experimentelle Mathematik,
             Universit\"at Duisburg-Essen,
             Ellernstrasse 29,
             45326 Essen, Germany}
\date{\today}

\maketitle

%\vspace{2mm} \begin{center}{\LARGE D R A F T  \quad  No 6-4} 
 %       \end{center}

\begin{abstract}
In this paper, we study $w$-frameproof codes, which are equivalent
to $\{1,w\}$-separating hash families. Our main results concern
binary codes, which are defined over an alphabet of two symbols.
For all $w \geq 3$, and for $w+1 \leq N \leq 3w$, we show that
an $\SHF(N; n,2, \{1,w \})$ exists only if $n \leq N$, and an
$\SHF(N; N,2, \{1,w \})$ must be a permutation matrix 
of degree $N$.
\end{abstract}

\section{Introduction}
Let $Q$ be a finite set of size $q$ and let $N$ be 
a positive integer. A subset $ C \subseteq Q^N$ with
$|C|=n$ is called $C$ an $(N,n,q)$ code. The elements of $C$
are called codewords.  Each codeword $x \in C$ is of the form
$x=(x_1, \ldots, x_N)$, where $x_i \in Q$, $1 \leq i \leq N$.  
For any subset of codewords $ P \subseteq  C$,
 the set of {\em descendants} of  $P$, denoted
 $\mbox{desc}(P)$, is defined by
 \[\mbox{desc}(P) = \{ x\in Q^N: x_i \in \{a_i: a 
  \in P \}, \; 1\leq i\leq N \}. \] 
Let $C$ be an $(N,n,q)$ code and let $w \geq 2$ be an integer.
$C$ is called a {\em w-frameproof code} ($w-\FPC$) 
if for all $P \subseteq C$
with $|P| \leq w$, we have that $\mbox{desc}(P) \cap C=P$. 
Frameproof codes with were first introduced by Boneh and Shaw 
\cite{Boneh1998},
for use in fingerprinting of digital data to prevent a small
of coalition of legitimate users from constructing a copy of
fingerprint of another user not in the coalition. 
Frameproof codes
and their applications
have been studied extensively, see for instance, 
\cite{Boneh1998}, \cite{FiatTassa1999}, 
\cite{ChorFiatNaor1994}, \cite{StaddonStinsonWei2001}, 
\cite{STW2000}, \cite{SarkarStinson2001}, \cite{Blackburn2003},
\cite{Colbourn2010}. 
 One of the basic problems
is the studying of upper bounds on the cardinality of frameproof codes.
Many strong bounds have been obtained in the papers 
\cite{StaddonStinsonWei2001}, \cite{SarkarStinson2001}, 
\cite{Blackburn2003}, \cite{TvT2013} for the case $q \geq w$.

Much less is known about upper bounds for frameproof codes 
when $q < w$. Our goal in the present paper is to study
upper bounds for binary frameproof codes, i.e., codes for  $q=2$.
 
It turns out that frameproof codes are a special
type of separating hash families (\SHF).
Let $h$ be a function from a set $X$ to a set $Y$ and let 
$C_1, C_2, \ldots, C_t \subseteq X$ be $t$ pairwise disjoint
subsets. We say that
$h$ {\it separates} $C_1, C_2, \ldots, C_t$ if 
$h(C_1), h(C_2), \ldots, h(C_t)$ are pairwise disjoint.
Let $|X|=n$ and $|Y|=q$. We call a set $\CH$ of $N$ functions from
$X$ to $Y$ an 
$(N; n,q, \{w_1, \ldots, w_t \})$-{\it separating hash family},
denoted by $\SHF(N; n,q, \{w_1, \ldots, w_t \})$, 
if for all 
pairwise disjoint subsets $C_1, \ldots, C_t \subseteq X$
with $|C_i|=w_i$, for $i=1, \ldots, t$, there exists at least
one function $h\in \CH$ that separates  $C_1, C_2, \ldots, C_t$.
The multiset  $\{w_1, w_2, \ldots, w_t \}$ is the {\it type}
of the separating hash family. 
Frameproof codes and separating hash families
have the following connection. An $(N,n,q)$ $w$-frameproof codes
exists if and only if an $\SHF(N; n,q, \{1, w \})$ exists.  
We include a lemma in section 2 proving this statement. As it
is more convenient to work with separating hash families,
we will prove the results in this paper in terms of separating
hash families.

It is often useful to present an 
$\SHF(N; n,q, \{w_1, \ldots, w_t \})$ as an $N\times n$ matrix on
$q$ symbols, say $\A$. The rows of $\A$ correspond to the hash
functions in the family, the columns correspond to the elements
in the domain $X$, and the entry in row $f$ and column $x$
is $f(x)$. We call $\A$ the matrix representation of the hash family. 
The matrix $\A$ has the following property. For given disjoint sets
of columns $C_1, C_2, \ldots, C_t$ with $|C_i|=w_i$, $1\leq i\leq t$,
there exists at least one row $f$ of $\A$ such that 
\[\{ \A(f,x): \; x \in C_i\} \cap \{ \A(f,x): \; x \in C_j\}
 = \emptyset , \]
for all $i\not=j$, i.e. row $f$ separates the column sets
$C_1, C_2, \ldots, C_t$.
Now if we write the codewords of an $(N,n,q)$ w-frameproof code
column-wise as an $N\times n$ matrix $\A$, i.e. each
codeword is a column of $\A$, then $\A$ is the matrix
representation of an $\SHF(N;n,q, \{1,w\})$. 
The problem of determining an upper bound on the cardinality
of an $(N,n,q)$ $w$-frameproof code becomes the problem
of determining an upper bound on the number of columns of $A$
for given $N$, $q$, and $w$.
 
For the case when $q \geq w$, several strong results 
have been obtained for $w$-frameproof codes. For example,
when $N \leq w$, it has been shown that $n\leq w(q-1)$, see
\cite{StaddonStinsonWei2001}, \cite{Blackburn2003}.
When $N > w$, strong upper bounds are obtained in
\cite{StaddonStinsonWei2001}, \cite{Blackburn2003},
\cite{BazTran2011}, \cite{TvT2013}. 
Here are these bounds.

\begin{thm}[\cite{StaddonStinsonWei2001}]\label{SSW_bound}
In an $(N,n,q)$ $w$-frameproof code, the following bound holds:
 $$ n \leq w(q^{\lceil {N\over w}\rceil}-1).$$
\end{thm}

\begin{thm}[\cite{Blackburn2003}]\label{Blackburn_bound}
Let $N$, $q$, $w$ and $d$ be positive integers such that
$N=wd+1$, $w\geq 2$ and $q\geq w$. 
Suppose there is an $(N,n,q)$ $w$-frameproof
code. Then $n\leq q^{d+1}+ O(q^d)$.
\end{thm}

\begin{thm}[\cite{TvT2013}] \label{Tran_bound}
Let $d$, $q$, $w$ be positive integers such that
$q \geq w \geq 2$. Suppose there exists an $(N,n,q)$ $w$-frameproof 
code with $N=wd+1$. Then $ n \leq q^{d+1}.$  
\end{thm} 

It should be mentioned that the bound of Theorem \ref{Tran_bound} 
is tight. Note also that when $N=w$ the bound $n\leq w(q-1)$ is tight
as well.

\subsection{Outline of the paper}

In Section \ref{first-interval.sec}, we consider the cases
when $w \geq 3$, $w+1 \leq  N \leq 2w+1$.
In Section \ref{second-interval.sec}, we consider the cases
when $w \geq 4$, $2w+2 \leq  N \leq 3w$.
Section \ref{exceptions.sec} handles the cases $w=3$, $N=8$ and $9$, which were
omitted from the previous section.
Section \ref{w=2.sec} briefly discusses the case $w=2$, and
Section \ref{conclusion.sec} is a conclusion.

\section{Bounds for binary $\FPC$ with $w+1 \leq  N \leq 2w+1$}  
\label{first-interval.sec}    

For the sake of completeness we include the following simple lemma.

\begin{lem}\label{FPC-SHF}
 An $(N,n,q)$ w-frameproof code is equivalent to an
 $\SHF(N;n,q, \{1,w\})$.
\end{lem} 
\begin{proof}
Let $\A$ be an $N\times n$ matrix having entries from a set
of $q$ symbols. Let $\{c\}$ and $P$ be any given disjoint subsets of 
columns of $\A$ with $|\{c\}|=1$ and  $|P|\leq w$, where $w$
is an integer such that $w\geq 2$. 
We may view $\A$ as an $(N,n,q)$ code whose codewords are
the columns. Assume that $\A$ is an $(N,n,q)$ w-frameproof code.
This is equivalent to say $\mbox{desc}(P) \cap \A = P$. Further, 
$\mbox{desc}(P) \cap \A = P$ is equivalent to the
statement that there is a row $i$ that separates $\{c\}$ and $P$.
The latter says that $\A$ is the matrix representation
of an $\SHF(N;n,q,\{1,w\})$.  
\end{proof}
  
By using Lemma \ref{FPC-SHF} we will prove the results
in terms of separating hash families.

When $q < w$,  
the statement of Theorem \ref{Tran_bound} is no longer valid.
The following construction gives a counter example to Theorem 
\ref{Tran_bound} when $q < w$. 
Let $N$, $q$, $w$ be positive integers such that
$q \geq 2 $. Let $\{0,1, \ldots , q-1\}$ be the symbol set.
Define an $ N  \times N(q-1)$ matrix $\A$
as follows.
\begin{center}
 \[\hspace{-12mm} \A = \hspace{6mm}  
  \overbrace{\hspace{-5mm} \left(     
  \begin{tabular}{ccccccccccccccc}
 1 & $\cdots$ & $q-1$ & 0 & $\cdots$ & 0 & 0 & $\cdots$ & 0 & 0 & $\cdots$ & 0
      & 0 & $\cdots$ & 0  \\
 0 & $\cdots$ & 0 & 1 & $\cdots$ & $q-1$ & 0 & $\cdots$ & 0 & 0 & $\cdots$ & 0  & 0 & $\cdots$ & 0  \\
 $\vdots$ & $\ddots$ & $\vdots$ & $\vdots$ & $\ddots$ & $\vdots$
  & $\vdots$ & $\ddots$ & $\vdots$ & $\vdots$ & $\ddots$ & $\vdots$
  & $\vdots$ & $\ddots$ &  $\vdots$  \\
 0 & $\cdots$ & 0 & 0 & $\cdots$ & 0 & 0 & $\cdots$ & 0 & 1 & $\cdots$ & $q-1$  &  0 & $\cdots$ & 0 \\
 0 & $\cdots$ & 0 & 0 & $\cdots$ & 0 & 0 & $\cdots$ & 0 & 0 & $\cdots$ & 0
 & 1 & $\cdots$ & $q-1$ \\
\end{tabular} \right ) \hspace{-5mm} }^{N(q-1)}  \hspace{5mm} 
\begin{tabular}{c}
\vspace{2mm}\rotatebox{270}{$\overbrace{ \hspace{25mm} }$ }  
 \end{tabular}  
\begin{tabular}{c}
  $N$
\end{tabular}
 \]       
\end{center} 
The matrix $\A$ has the property that for any given column $c$
there exists a row $r$ such that the entry $\A(r,c)$ is unique.
Hence $\A$ is the matrix representation of an 
$\SHF(N; N(q-1), q, \{1,w\})$ for $w \geq 1$. 
Note that this construction
can be found in Blackburn \cite{Blackburn2003} using another 
description, as $q$-ary codes of length $N$ whose all codewords 
have weight exactly 1. If we choose, for example, $w=N-1$ and 
$q < \sqrt{w}$, then we have
 $N(q-1)=(w+1)(q-1) > (q^2+1)(q-1)> q^2$.
Thus $A$ is a counter example to
Theorem \ref{Tran_bound}, that would yield $n \leq q^2$ for this case.

Finding a tight bound for $w$-frameproof codes
with $q < w$ seems to be a challenging problem. In the following 
we focus on the case $q=2$ and prove certain tight bounds 
for binary $w$-frameproof codes 
when their length $N$ is moderate compared to $w$.

For any given frameproof code, we may derive new ones from it by simply
permuting the entries in each row separately, i.e. a permutation of the
elements $1,\ldots,q$. Such codes can be considered to be in the same
equivalence class, and hence we would like to limit ourselves to considering
a fixed representative. In the binary case, we say that an $\SHF(N;n,2,{1,w})$
is in {\it standard form} if every row has at most $n/2$ entries of 1.

We now record a simple fact about the binomial coefficients.

\begin{lem}\label{binomial}
Let $w$, $n$ be positive integers such that $w+1 \leq n$. Then for
$i = 1,2,\ldots,n-w-1$, we have $i{n-i \choose w} > (i+1){n-i-1 \choose w}$
if and only if $(i+1)(w+1) > n+1$. In particular, we have
\begin{eqnarray}\label{eq2}
{n-1 \choose w} > 2 {n-2 \choose w} > 3 {n-3 \choose w} >
 \cdots > j {n-j \choose w}. 
\end{eqnarray}
for $j \leq n-w$ whenever $n \leq 2w$.
\end{lem}

\begin{proof}
\begin{align*}
i{n-i \choose w} > (i+1){n-i-1 \choose w}
&\Leftrightarrow \frac{i (n-i)!}{(n-i-w)! \cdot w!} > \frac{(i+1)(n-i-1)!}{(n-i-w-1)! \cdot w!} \\
&\Leftrightarrow \frac{i (n-i)}{(n-i-w)} > i+1 \\
&\Leftrightarrow ni - i^2 > ni + n - i^2 - i - iw - w \\
&\Leftrightarrow i + iw + w > n \\
&\Leftrightarrow (i+1)(w+1) > n+1
\end{align*}
Note that Equation \ref{eq2} holds if and only if
$i{n-i \choose w} > (i+1){n-i-1 \choose w}$ holds for
$i=1$, which corresponds to $2(w+1) > n+1$ or equivalently
$n \leq 2w$. 
\end{proof}

%\vspace{4mm}

We introduce some definitions. Let $\A$ be the representation matrix
of an $\SHF(N;n, 2, \{1,w\})$. A row $r$ of $\A$ is said to be of
{\it type i} if $r$ contains exactly $i$ entries 1. Two rows $r_1$ and
$r_2$ of $\A$ are said to be {\it overlapped} if they share a column
in which both rows have an entry 1. If rows $r_1$ and $r_2$ are not
overlapped, we say that they are {\it disjoint}.

For an arbitrary $\SHF(N;n,2,{1,w})$ $\A$, it is clear that both 0 and 1 have
to occur in each row of $\A$, otherwise that row would not contribute to the
separation of any pair $(C_1, C_2)$. Hence we may assume that $\A$ contains no
row of type 0 in standard form, by simply removing any such row and replacing
them with an arbitrary row of type 1.

The following observation will be used throughout this paper.

\begin{lem}
\label{separate.lem}
Let $\A$ be an $\SHF(N;n,2,{1,w})$. Suppose row $r$ of $\A$ is of type $i \leq n/2$.
If $i < w$, then row $r$ separates exactly $i{n -i \choose w}$ column pairs $(C_1, C_2)$.
If $i \geq w$, then row $r$ separates exactly $i{n -i \choose w}+ {i \choose w}(n-i)$
column pairs $(C_1,C_2)$.
\end{lem}

We will now prove a bound for binary frameproof codes.

\begin{thm}\label{bound_1}
Let $w$, $N$ be positive integers such that $ w \geq 3$
and $w+1 \leq N \leq 2w+1$. Suppose there exists an
$\SHF(N; n, 2, \{1,w\})$. Then  
$n \leq N.$
\end{thm}

\begin{proof}
Suppose, by contradiction, that there exists an 
$\SHF(N; n, 2, \{1,w\})$
with $n=N+1$. Let $\A$ be its $N\times (N+1)$ 
matrix representation on the symbol set 
$\{0,1\}$. Let $\T$ be the total number of pairs of 
disjoint column sets $(C_1, C_2)$
of $\A$ with $|C_1|=1$ and $|C_2|=w$ that need to be separated.
Then we have $\T:={n \choose w}(n-w)=n{n-1 \choose w}$.

Consider the following three cases regarding the number of columns of $\A$.
\begin{enumerate}
\item[(i)] $n=N+1 \leq 2w$ (i.e. $N \leq 2w-1$).

Using Lemma \ref{binomial} we see that
$${n-1 \choose w} > 2 {n-2 \choose w}> 3{n-3 \choose w}> \cdots >
 (w-1){n -(w-1) \choose w} . $$
The term $j {n-j \choose w}$ in these inequalities corresponds to the 
number of column pairs $(C_1, C_2)$ separated by a row of type $j$.
Hence a row of type 1 separates
the largest number of column pairs $(C_1, C_2)$, namely
${n-1 \choose w}= {N \choose w}$. Moreover, since $\A$ has
$N$ rows, the maximal number of column pairs $(C_1, C_2)$ that
can be separated by all the rows of $\A$ is therefore
$N{N \choose w} =(n-1){n-1 \choose w}$. This is a contradiction,
since $(n-1){n-1 \choose w} < \T.$

\item[(ii)] $n=N+1=2w+1$ (i.e. $N =2w$).

Observe that we have
$${n-1 \choose w} =
{N \choose w}={2w \choose w} = 2{2w-1 \choose w}=2{n-2 \choose w}$$
in this case. This observation together with Lemma \ref{binomial}
give rise to the following inequalities about the number
of column pairs $(C_1, C_2)$ separated by a row of type $j$,
where $j=1, \ldots , w$.
$${n-1 \choose w} = 2 {n-2 \choose w}> 3{n-3 \choose w}> \cdots >
 (w-1){n -(w-1) \choose w} > w{n-w \choose w}+ {n-w}. $$
The last inequality can be easily checked, while all other inequalities
follow from Lemma \ref{binomial} Note that the last term of the
inequalities corresponds to the case of a row of type $w$.
Again, this implies that a row of $\A$ can separate at most
${n-1 \choose w}={N \choose w}$ column pairs $(C_1, C_2)$. 
Thus all $N$ rows of $\A$
can separate at most $N{N \choose w}=2w{2w\choose w}$
column pairs $(C_1, C_2)$, whereas the total number of column pairs 
$(C_1, C_2)$ that need to be separated  is 
$\T={N+1 \choose w}(N+1-w)=(2w+1){2w\choose w}$, a contradiction.

\item[(iii)] $n=N+1=2w+2$ (i.e. $N =2w+1$).

In this case we have the following inequalities
$$2{2w \choose w} > {2w+1 \choose w} >
3{2w-1 \choose w}> \cdots >
(w-1){w+3 \choose w}> w{w+2 \choose w}+(w+2) > 2(w+1)^2. $$
The last two inequalities can be easily checked, while the other
inequalities follow from Lemma \ref{binomial}. Here the first term
of the inequalities corresponds to a row of type 2; the second term
to a row of type 1; the third term to a row of type 3, etc., the last
term corresponds to a row of type $\lfloor {n/2} \rfloor =(w+1)$.

Recall that the total number of column pairs $(C_1, C_2)$
is $\T={n \choose w}(n-w)= {2w+2 \choose w}(w+2)$.
We show that if each row of $\A$ separates a maximal number
of column pairs $(C_1, C_2)$, then all the $N=2w+1$ rows of $\A$ fail
to separate all $\T$ column pairs $(C_1, C_2)$. In fact, this corresponds
to the first term of the above inequalities. This is the case
for which each row of $\A$ is of type 2. So each row
will separate $2{2w \choose w}$ column pairs $(C_1, C_2)$. Hence all
$N=2w+1$ rows of $\A$ will separate at most
 $$\Z:= 2(2w+1){2w \choose w}$$
column pairs $(C_1, C_2)$ of $\A$. 
Now using the equality ${n \choose m} = {\frac{n}{n-m}}{n-1 \choose m}$
we see that 
$$\T={2w+2 \choose w}(w+2) ={\frac{(2w+2)}{(w+2)}}{\frac{(2w+1)}{(w+1)}}(w+2)
   {2w \choose w}=2 (2w+1){2w \choose w} = \Z.$$

However, if each row of $\A$ is of type 2, then there must exist
two overlapped rows, say $r_1$ and $r_2$. 
These rows $r_1$ and $r_2$
will then separate 
${2w-1 \choose w}$ common column pairs $(C_1,C_2)$. 
This leads to a contradiction, since all the rows of $\A$ will separate
less than $\T$ column pairs $(C_1,C_2)$. This completes the proof.  
\end{enumerate}
\end{proof}

%\vspace{4mm}

Recall that a binary $N\times N$ matrix $\A$ is called
a {\em permutation matrix} of degree $N$ if $\A$ has
precisely one entry equal to 1 in each row and each
column, and 0s elsewhere. 
It is obvious that any permutation matrix
of degree $N$ is the representation matrix of
an $\SHF(N;N, 2 \{1, w\})$  for any $w \leq N-1$. Hence,
the bound of Theorem \ref{bound_1} is tight.
In the following, we prove a stronger result
which states that permutation matrices are
the only solutions for an  $\SHF(N;N, 2, \{1, w\})$ 
with  $w+1 \leq N \leq 2w+1$ and $ w \geq 3$.

\begin{thm}\label{tight_1}
Let $w$, $N$ be positive integers such that $ w \geq 3$
and $w+1 \leq N \leq 2w+1$. Suppose there exists an
$\SHF(N; n, 2, \{1,w\})$ with $n=N$. Then its representation
matrix in standard form is a permutation matrix of degree $N$.
\end{thm}

\begin{proof}
Let $\A$ be the representation matrix of an $\SHF(N; N, 2, \{1,w\})$
in standard form with $w+1 \leq N \leq 2w+1$ and
$w \geq 3$. Consider two cases.

\begin{enumerate}
\item[(i)] $n=N\leq 2w$.

Recall that the total number of column pairs
$(C_1, C_2)$ of $\A$ that need to be separated is 
$\T={N \choose w}(N-w)$. By Lemma \ref{binomial}
each row of $\A$ can separate at most ${N-1 \choose w}$
column pairs $(C_1, C_2)$, and this case occurs when 
each row is of type 1. Thus the largest number
of separated column pairs $(C_1, C_2)$ obtained by $N$ rows of $\A$ is
$N{N-1 \choose w}={N \choose w}(N-w)$. This number is achieved 
if and only if the unique entries 1 of the rows belong to the different
columns, i.e., $\A$ is a permutation matrix of degree $N$.

\item[(ii)] $n=N=2w+1$.

In this case we have $\T={2w+1 \choose w}(w+1)$.
A row $r$ of $\A$ can separate at most $2w \choose w$ column pairs
$(C_1,C_2)$. This number corresponds to $r$ being
of either type 1 or type 2. Further, 
the maximum number of possible separated column pairs $(C_1, C_2)$
which may be achieved by all the rows of $\A$ is
$(2w+1){2w \choose w}$. 
To achieve the
maximum number $(2w+1){2w \choose w}$ of separated column pairs,
any two rows of $\A$ have to separate disjoint sets of 
column pairs $(C_1, C_2)$. This implies that any two rows of $\A$
are disjoint. 
This is equivalent to saying that
each column of $\A$ contains exactly one entry 1,
otherwise if two rows $r_1$ and $r_2$ are overlapped, then these two rows separate
a common non-empty subset of column pairs $(C_1, C_2)$, which is
a contradiction. Therefore, $\A$ is a permutation matrix of
degree $2w+1$.    
\end{enumerate}
\end{proof}

\section{ Bounds for binary $\FPC$ with $w \geq 4$ and $2w+2 \leq N \leq 3w$ }
\label{second-interval.sec}

In this section, we present a result that allows characterization of $\SHF(N;N,2,{1,w})$
for $w \geq 4$ and $N \leq 3w$. In particular, we prove that all such separating hash
families in standard form are permutation matrices. This type of result allows us to prove
bounds similar to Theorem \ref{bound_1} by using the following theorem.

\begin{thm} \label{perm_bound}
Let $w \geq 3$, $N \geq w+1$ and suppose that all $\SHF(N;N,2,\{1,w\})$ in standard form
are permutation matrices. If $\SHF(N;n,2,\{1,w\})$ exists, then $n \leq N$.
\end{thm}

\begin{proof}
Suppose not, then there exists some $\SHF(N;N+1,2,\{1,w\})$, say $\A$. Let $\B$ be the
submatrix formed by the first $N$ columns of $\A$. We may assume w.l.o.g. that $\B$
is in standard form (we may need to permute 0s and 1s in each row of $\A$ to achieve this).
Then $\B$ is a permutation matrix. Thus each row of $\A$ has at most two entries of 1.

Since $N \geq w+1 \geq 4$, we have that $N/2 \geq 2$. Let $\C$ be the submatrix formed by
the last $N$ columns of $\A$. Each row of $\C$ has at most two entries of 1 as well, so
$\C$ is in standard form, and hence it is a permutation matrix. This implies the first and
last columns of $\A$ are identical, which is a contradiction since $(\{1\},\{N+1\})$ cannot
be separated. 
\end{proof}

%\vspace{4mm}
We may use Theorem \ref{perm_bound} to give a second proof of Theorem \ref{bound_1} using
Theorem \ref{tight_1}. In light of this result, it is also important to consider the question
``when are permutation matrices the only representatives of $\SHF(N;N,2,{1,w})$ in standard form?''
We give an affirmative answer for $w \geq 4$ and $N \leq 3w$ through a series of lemmas below.

\begin{lem} \label{type_1}
Let $w \geq 3$ and let $\A$ be the representation matrix of an $\SHF(N;N,2,\{1,w\})$.
Suppose that all $\SHF(N-1;N-1,2,\{1,w\})$ in standard form are permutation matrices.
If $\A$ contains a row of type 1, then $\A$ is a permutation matrix.
\end{lem}

\begin{proof}
We can write $\A$ in the form
\begin{displaymath}
\A =
\left( \begin{array}{c|ccc}
    1  & 0 & \ldots & 0 \\ \hline
            &    &         &   \\
             &    & \B      &   \\
             &    &         &   \\    
    \end{array} \right)
\end{displaymath}
  
Let $\B$ be the $(N-1) \times (N-1)$ matrix obtained from $\A$
by removing the first row and  the first column of $\A$. 
Then $\B$ is the representation matrix of an $\SHF(N-1; N-1, 2,\{1,w\})$.
We may assume w.l.o.g. that $\B$ is in standard form, and hence it is a
permutation matrix.

By permuting the columns of $\A$, if necessary, we may assume that $\B$
is the identity matrix. Consider column pairs $(C_{x}=\{x\}, C_{1,y,z}=\{1,y,z\})$
with $x,y,z =\{2, \ldots, N\}$ and $x \not= y \not=z \not=x$.
Since $\B$ is the identity matrix,
a row that separates $(C_{x}, C_{1,y,z})$ must have entry 0 in columns
$1,y,z$ and entry 1 in column $x$. Thus row $x$ is the unique row
separating $(C_{x}, C_{1,y,z})$. It follows that $\A$ is a permutation
matrix. 
\end{proof}

\begin{lem} \label{ext_shf}
Let $w \geq 4$, $N \leq 3w$, and let $\A$ be the representation matrix of
an $\SHF(N;N,2,\{1,w\})$. Suppose the first row of $\A$ is of type $i_0 \leq w$ with
$\A(1,1) = 1$. Let $\B$ be the submatrix by deleting the first row and first column of
$\A$. Then $\B$ is an $\SHF(N-1;N-1,2,\{1,w\})$.
\end{lem}

\begin{proof}
If $\A$ contains a row of type 1 then Lemma \ref{type_1} applies. For the remainder
of this proof, we assume that $\A$ contains no row of type 1.

Suppose $\B$ is not an $\SHF(N-1;N-1,2,\{1,w\})$, then there exists some column set
pair $(C_1 = \{x\},C_2)$ with $|C_2| = w$ that cannot be separated by $\B$. If $x$
corresponds to a column of $\A$ that has an entry of 0 in the first row then $C_2$
contains a column of $\A$ that has an entry of 0 in the first row since $i_0 - 1 < w$.
But then $\A$ also cannot separate $(C_1,C_2)$; a contradiction. Thus $x$ contains a
1 in the first row, and all columns of $C_2$ correspond to columns of $\A$ with 0's in
the first row (otherwise $\A$ still cannot separate $(C_1,C_2)$).

Permute the columns of $\A$ so that $x$ corresponds to column 2 and columns in $C_2$
correspond to columns $3,\ldots,w+2$. The matrix $\A$ is now

\begin{displaymath}
\A =
\left( \begin{array}{ccccc|c}
     1  & 1 & 0 & \cdots & 0  & \\ \hline
        &     &     &     &      & \\
        &     &     &     &      & \\  
    \end{array} \right)
\end{displaymath}

For $1 \leq i \leq w$, let $C_i = \{3,\ldots,w+2\} \setminus \{i+2\}$. The column
set pair $(\{2\}, C_i \cup \{1\})$ must be separated by $\A$. By permuting 0's and 1's
if necessary, there is some row $r_i \neq 1$ with entry 1 in column 2 and entry 0 in
columns of $C_i$. Since $C_i \cup C_j = \{3,\ldots,w+2\}$ for $i \neq j$ and $\B$ does
not separate $(\{2\}, \{3,\ldots,w+2\})$, we have that $r_i \neq r_j$ for $i \neq j$.
Moreover, entry $i$ of $r_i$ must also be a 1. Let $R_1 = \{r_1,\ldots,r_w\}$, and by
permuting the rows of $\A$ we have

\begin{displaymath}
\A =
\left( \begin{array}{ccccccc|c}
     1 & 1 & 0 & 0 & 0 & \cdots & 0  & \\ \hline
     0 & 1 & 1 & 0 & 0 & \cdots & 0  & \\
     0 & 1 & 0 & 1 & 0 & \cdots & 0  & \\  
     \vdots & \vdots &  &  &  & \ddots & \vdots  & \\  
     0 & 1 & 0 & 0 & 0 & \cdots & 1  & \\
     * & * & * & * & * & \cdots & *  & \\
     \vdots & \vdots &  &  &  & \ddots & \vdots  & \\
     * & * & * & * & * & \cdots & *  & \\  
    \end{array} \right)
\end{displaymath}

Next, consider $C'_i = \{2,\ldots,w+2\} \setminus \{i+2\}$ for $i = 1,\ldots,w$.
The column set pair $(\{i+2\}, C'_i)$ must be separated by $\A$ with some row
$r'_i \neq 1$ and $r'_i \notin R_1$. By permuting the 0's and 1's if necessary,
$r'_i$ has entry 1 in column $(i+2)$ and entry 0 in columns in $C'_i$. Moreover,
$r'_i \neq r'_j$ for $i \neq j$. Now let $R_2 = \{r'_1,\ldots,r'_w\}$, and by
permuting the rows of $\A$ we have

\begin{displaymath}
\A =
\left( \begin{array}{ccccccc|c}
     1 & 1 & 0 & 0 & 0 & \cdots & 0  & \\ \hline
     0 & 1 & 1 & 0 & 0 & \cdots & 0  & \\
     0 & 1 & 0 & 1 & 0 & \cdots & 0  & \\  
     \vdots & \vdots &  &  &  & \ddots & \vdots  & \\  
     0 & 1 & 0 & 0 & 0 & \cdots & 1  & \\ \hline
     * & 0 & 1 & 0 & 0 & \cdots & 0  & \\
     * & 0 & 0 & 1 & 0 & \cdots & 0  & \\  
     \vdots & \vdots &  &  &  & \ddots & \vdots  & \\  
     * & 0 & 0 & 0 & 0 & \cdots & 1  & \\ \hline
      * & * & * & * & * & \cdots & *  & \\
     \vdots & \vdots &  &  &  & \ddots & \vdots  & \\
     * & * & * & * & * & \cdots & *  & \\  
    \end{array} \right)
\end{displaymath}

We now do the following addition of rows in steps, starting with $R_3 = \emptyset$:

\begin{description}
\item[Step 1]\mbox{\quad}\\
Let $a$ be the column 1 entry of $r'_1$. If $a = 1$, consider the column 
pair $(\{3\}, \{1,\ldots,w+1\} \setminus \{3\})$, which must be separated by
some row $r''_1 \neq 1$ of $\A$. Note that $r''_1 \notin R_1$ and $r''_1 \notin R_2$.
Add $r''_1$ to $R_3$.

If $a = 0$, consider the column  pairs $(\{3\}, C''_{1,j} = \{2,4,5,w+j+2\})$ for
$j = 1,\ldots,N-w-2$. Since $w \geq 4$, we have that $\A$ separates $(\{3\}, C''_{1,j})$.
If $r'_1$ separates every such pair then $r'_1$ is a type 1 row;
a contradiction to $\A$ having no type 1 rows. Thus there is some $j$ such that another
row of $\A$, call it again $r''_1$, that separates $(\{3\}, C''_{1,j})$. Note that
$r''_1 \neq 1$, $r''_1 \notin R_1$ and $r''_2 \notin R_2$. Add $r''_1$ to $R_3$.

\item[Step 2]\mbox{\quad}\\
Let $a$ be the column 1 entry of $r'_2$. If $a = 1$, consider the column  pair
$(\{4\}, \{1,\ldots,w+1\} \setminus \{4\})$, which must be separated by some row
$r''_2 \neq 1$ of $\A$. Note that $r''_2 \notin R_1 \cup R_2$ and $r''_2 \neq r''_1$.
Add $r''_2$ to $R_3$.

If $a = 0$, consider the column  pairs $(\{4\}, C''_{2,j} = \{2,3,5,w+j+2\})$ for
$j = 1,\ldots,N-w-2$. Similar to Step 1, there exists some $j$ for which another row
of $\A$, call it again $r''_2$, that separates $(\{4\}, C''_{2,j})$. Again
$r''_2 \notin R_1 \cup R_2$ and $r''_2 \neq r''_1$. Add $r''_2$ to $R_3$.

\item[Steps $i = 3,\ldots,w-1$]\mbox{\quad}\\
Let $a$ be the column 1 entry of $r'_i$. If $a = 1$, consider the column  pair
$(\{i+2\}, \{1,\ldots,w+1\} \setminus \{i+2\})$, which must be separated by some
row $r''_i \neq 1$ of $\A$. Note that $r''_i \notin R_1 \cup R_2 \cup R_3$. Add
$r''_i$ to $R_3$.

If $a = 0$, consider the column  pairs $(\{i+2\}, C''_{i,j} = \{2,3,\ldots,i+1,w+j+2\})$
for $j = 1,\ldots,N-w-2$. Since $|C''_{i,j}| = i+1 \leq w$, some row of $\A$ separates
$(\{i+2\}, C''_{i,j})$. Similar to Step 1, there exists some $j$ for which another row of
$\A$, call it again $r''_i$, that separates $(\{i+2\}, C''_{i,j})$. Again
$r''_i \notin R_1 \cup R_2 \cup R_3$. Add $r''_i$ to $R_3$.

\item[Step $w$]\mbox{\quad}\\
Consider the column set pair $(\{1\},\{2,\ldots,w+1\})$, which must be separated
by some row $r$ of $\A$. Clearly $r \notin R_1 \cup R_2 \cup R_3$. Add $r$ to $R_3$.
\end{description}

At the end of Step $w$, we have added $w$ distinct rows to $R_3$, so $\A$ has at
least $|R_1 \cup R_2 \cup R_3| + 1 = w + w + w + 1 = 3w+1$ rows. This contradicts
$N \leq 3w$, so Lemma \ref{ext_shf} holds. 
\end{proof}

\begin{lem} \label{extension}
Let $w \geq 4$, $w+1 \leq N \leq 3w$, and let $\A$ be the representation matrix of
an $\SHF(N;N,2,\{1,w\})$. Suppose that some row of $\A$ is of type at most
$w$ and all $\SHF(N-1;N-1,2,\{1,w\})$ in standard form are permutation matrices.
Then $\A$ is a permutation matrix.
\end{lem}

\begin{proof}
If $\A$ contains a row of type 1, we can use Lemma \ref{type_1} to show that
$\A$ is a permutation matrix. For the remainder of this proof, we may assume that
$\A$ contains no row of type 1. Assume w.l.o.g. that the first row of $\A$ is of
type $i_0$ where $2 \leq i_0 \leq w$.

Suppose to the contrary that $\A$ is not a permutation matrix. By permuting the
columns of $\A$ if necessary, we may assume that row 1 is $1^{i_0}0^{N-i_0}$.
Let $\B$ be the $(N-1) \times (N-1)$ submatrix of $\A$ by deleting the first row
and first column of $\A$.

By Lemma \ref{ext_shf}, we have that $\B$ is an $\SHF(N-1;N-1,2,\{1,w\})$, and
hence it is a permutation matrix. For row $x$ of $\A$, $x = 2,\ldots,N$, let $c_x$
be the unique column of $\A$ that contains a 1 in row $x$. Consider the column
set pair $(C_x = \{c_x\}, C'_{x} = C''_{x} \cup \{1\})$ where $C''_{x}$ is some
set of $w-1$ columns not containing $c_x$ whose entries on row 1 contains at least
one 0. This is possible since $N \geq w+2 \geq i_0 + 2$. The only row that can
separate this column set pair is row $x$, which forces its first entry to be a 0.
Thus we have
shown that
\begin{displaymath}
\A =
\left( \begin{array}{c|ccc}
     1   & 1 &  &  \\ \hline
     0   &     &     &   \\
     \vdots   &     & \B  &   \\
     0   &     &     &   \\    
    \end{array} \right).
\end{displaymath}
Now consider $(C_1 = \{1\}, C_2 = \{2,3\})$, which cannot be separated by $\A$;
a contradiction. 
\end{proof}

\begin{thm} \label{tight_2}
Let $w$, $N$ be positive integers such that $ w \geq 4$
and $2w+2 \leq N \leq 3w$. Suppose there exists an
$\SHF(N; N, 2, \{1,w\})$. Then its representation
matrix in standard form is a permutation matrix of degree $N$.
\end{thm}

\begin{proof}
The proof is by induction on $N = 2w+1,\ldots,3w$. The base case $N = 2w+1$ is given
by Theorem \ref{tight_1}. Suppose that $N > 2w+1$ and all $\SHF(N-1;N-1,2,\{1,w\})$
in standard form are permutation matrices. By Lemma \ref{extension}, we only need to
show that some row of type at most $w$ exists.

Let $\A$ be an $\SHF(N;N,2,\{1,w\})$ in standard form. Fix some $i$ where
$w+1 \leq i \leq N/2$. The average number of column pairs separated by a row is
$$\alpha = \frac{(N-w){N \choose w}}{N} = {N-1 \choose w}.$$ Let $\beta_i$ be the
number of column pairs separated by a row of type $i$, then
$$\beta_i = i{N-i \choose w} + (N-i){i \choose w} \leq N{N-i \choose w}$$
column pairs. Since $i \geq w+1$, we have
\begin{align*}
\alpha &= {N-1 \choose w} \\
			 &= \frac{(N-1)(N-2)\cdots(N-w)}{(N-w-1)(N-w-2)\cdots(N-2w)}{N-w-1 \choose w} \\
			 &\geq \frac{(N-1)(N-2)\cdots(N-w)}{(N-w-1)(N-w-2)\cdots(N-2w)}{N-i \choose w} \\
			 &\geq \left( \frac{N-1}{N-w-1} \right)^w {N-i \choose w} \\
			 &\geq \left( \frac{3w+1-1}{3w+1-w-1} \right)^w {N-i \choose w} \\
			 &= \left( \frac{3}{2} \right)^w {N-i \choose w}.
\end{align*}
For $w \geq 8$, one can check that $\left( \frac{3}{2} \right)^w > 3w \geq N$, so
$\alpha > \beta_i$. It is straightforward to compute $\alpha$ and $\beta_i$ for
$4 \leq w \leq 7$ and confirm that $\alpha > \beta_i$ for all relevant values of $i$.
Since $\alpha > \beta_i$ for every $i \geq w+1$ and $\A$ contains no row of type
$N/2 + 1$ or higher, there must exist some row of type at most $w$. 
\end{proof}

Finally, we give a bound similar to Theorem \ref{bound_1}.

\begin{thm} \label{bound_2}
Let $w$, $N$ be positive integers such that $w \geq 4$
and $2w+2 \leq N \leq 3w$. Suppose there exists an
$\SHF(N; n, 2, \{1,w\})$. Then 
$n \leq N.$
\end{thm}

\begin{proof}
By Theorem \ref{tight_2}, all $\SHF(N;N,2,\{1,w\})$ in standard form are permutation
matrices, hence the proof follows from Theorem \ref{perm_bound}. 
\end{proof}

\section{ Binary $\FPC$ with $w=3$ and $N = 8 , 9$}
\label{exceptions.sec}

In this section we treat the cases $w=3$ when $N=2w+2=8$ and $N=3w=9$.
We show that Theorem \ref{tight_2} and 
Theorem \ref{bound_2} proven in the previous section
remain valid for $w=3$. The reason for a separate 
discussion of the case $w=3$ is that the proof for case $w \geq 4$
cannot be used for $w=3$.

\subsection{The case $w=3$ and $N=8$ }
We first consider the case of $N=8$.
Before we prove our main result,  we prove several useful lemmas
of a general nature.

Given two rows of an SHF, we define the {\it overlap} of the two rows to be the number
of columns in which both rows contain a 1.

\begin{lem} \label{overlapping}
 Let $\A$ be the representation matrix of an $\SHF(N; n, 2, \{1,w\})$.
 Let $r_i$ be a row of type $i$ and let $r_j$ be a row of type $j$
 of $\A$. Suppose that $r_i$ and $r_j$ have overlap equal to $s$.
 Then the number of column pairs $(C_1, C_2)$ that are separated by
 both of $r_i$
 and $r_j$ is
  \begin{equation}
  \label{overlap.eq}
   \theta = s{n-i-j+s \choose w}+ (n-i-j+s){s \choose w} 
   + (i-s){j-s \choose w}+ (j-s){i-s \choose w} .
   \end{equation} 
%In particular,
% \begin{itemize}
% \item [(i)]
%  if $i < w$, $j < w$, then $\theta =0$;
%\item [(ii)]
%  if $i < w$, $j \geq w$, then $\theta =i{j \choose w}$.
% \end{itemize}
\end{lem}

\begin{proof}
For $k, \ell \in \{0,1\}$, let $f(k, \ell)$
denote the set of columns in which $r_i$ has the entry $k$
and $r_j$ has the entry $\ell$.
Then $|f(1,1)| = s$,  $|f(1,0)| = i-s$, 
$|f(0,1)| = j-s$, and   $|f(0,0)| = n-i-j+s$. 
We have repeated column pairs $(C_1,C_2)$ in the following four situations:
\begin{enumerate}
\item $C_1 \subseteq f(1,1)$, $C_2 \subseteq f(0,0)$,
\item $C_1 \subseteq f(0,0)$, $C_2 \subseteq f(1,1)$,
\item $C_1 \subseteq f(1,0)$, $C_2 \subseteq f(0,1)$, and
\item $C_1 \subseteq f(0,1)$, $C_2 \subseteq f(1,0)$.
\end{enumerate}
These four cases correspond to the four summands in equation (\ref{overlap.eq}).
\end{proof}

In general, we will consider an SHF one row at a time.
Suppose the rows of an $\SHF(N; n, 2, \{1,w\})$ are denoted
$r_1, \dots , r_n$. For $1 \leq i \leq n$, define
$\mu_i$ to be the number of column pairs $(C_1, C_2)$
separated by $r_i$ that were not separated by $r_1, \dots , r_{i-2}$ or  $r_{i-1}$.

\begin{lem} \label{6r5c}
Let $\A$ be the representation matrix in standard form 
of an $\SHF(6; 5, 2, \{1,3\})$.
Then by permuting the rows of $\A$ we have that the first five rows
are of type 1 and the last row is of any type.
\end{lem}

\begin{proof}
The proof of the lemma is by straightforward counting.
First of all note that there are in total $T={5 \choose 3}2=20$ column pairs $(C_1,C_2)$ 
of $\A$ to be separated (where $|C_1|=1$, $|C_2|=3$ and
$C_1 \cap C_2= \emptyset$). On average each row separates
$\frac{20}{6} > 3$ new column pairs. 
From Lemma \ref{separate.lem}, a row of type 1
separates four column pairs and a row of type 2 separates
two column pairs. 
Suppose without loss of generality that $\mu_1 \geq \mu_2 \geq \dots \geq \mu_6$.
It follows that the first row is of type
1. Row 2 has to separate at least $\lceil \frac{20-4}{5} \rceil=4$
new column pairs, so row 2 is of type 1. Row 3 has to separate at least
$\lceil \frac{16-4}{4} \rceil=3$ column pairs, so row 3 is of type 1.
Row 4 also has to separate at least
$\lceil \frac{12-4}{3} \rceil=3$ column pairs, so row 4 is of type 1.
Now row 5 has to separate at least $\lceil \frac{8-4}{2} \rceil=2$ 
column pairs. If row 5 is of type 2, then it has to overlap at least
one of the first four rows (which are all of type 1). Then row 5 can
separate at most one new column pair, from Lemma \ref{overlapping}. It follows
that row 5 is also of type 1 and the first five rows separate all the column pairs.   
\end{proof}

\begin{lem} \label{7r6c}
Let $\A$ be the representation matrix in standard form 
of an $\SHF(7; 6, 2, \{1,3\})$
Then by permuting the rows of $\A$ we have that the first six rows
are of type 1 and the last row is of any type.
\end{lem}

\begin{proof}
There are in total $T={6 \choose 3}3=60$ column pairs $(C_1,C_2)$ 
of $\A$ to be separated (where $|C_1|=1$, $|C_2|=3$ and
$C_1 \cap C_2= \emptyset$). On average each row separates
$\frac{60}{7} > 8$ column pairs.
From Lemma \ref{separate.lem},  a row of type 1
separates ten column pairs, a row of type 2
separates eight column pairs, and a row of type 3
separates six column pairs. 
Suppose without loss of generality that $\mu_1 \geq \max \{\mu_2, \dots ,\mu_7\}$.
It follows
that the first row of $\A$ is of type 1. By permuting the 
columns if necessary, we assume that the first row 
has the 
entry 1 in the first column. Then $\A$  has the form
\begin{displaymath}
\A =
\left( \begin{array}{c|ccccc}
    1  & 0 & 0 & 0 & 0 & 0  \\ \hline
    a  &    &    &    &    &   \\
    b  &    &    &    &    &   \\
    c  &    &    &   \B    &   \\
    d  &    &    &    &    &   \\
    e  &    &    &    &    &   \\
    f  &    &    &    &    &     
    \end{array} \right)
\end{displaymath}
where $\B$ is the representation matrix of an $\SHF(6;5,2,\{1,3\})$.
By Lemma \ref{6r5c}, we may assume
\begin{displaymath}
\B =
\left( \begin{array}{ccccc}
    1 & 0 & 0 & 0 & 0  \\
    0 & 1 & 0 & 0 & 0 \\
    0 & 0 & 1 & 0 & 0 \\
    0 & 0 & 0 & 1 & 0 \\
    0 & 0 & 0 & 0 & 1 \\
    * & * & * & * & *
    \end{array} \right).
\end{displaymath}
So rows $2, \dots , 6$ have type $1$ or $2$.
If $a=b=c=d=e=0$, then the first six rows of $\A$ are of type 1 and they separate all 60 column pairs.
Suppose that not all of $a, b,c,d, e$ are 0.  
Then, from Lemma \ref{separate.lem} and Lemma \ref{overlapping},  the first six rows of $\A$ separate at most 
$5\times 10 + 4=54$ column pairs, and this occurs if and only if
there is exactly one nonzero element in $\{a, \dots , f\}$. It follows that row 7 has to separate
at least six new column pairs. This is impossible due to overlapping with rows 1 to 6,
unless row $7$ is of type 1. In this case we can interchange row 7 with one of the first
six rows to obtain the desired conclusion.
\end{proof}

\begin{thm}\label{tight_3}
The representation matrix $\A$ in standard form of an $\SHF(8;8,2,\{1,3\})$
is a permutation matrix of degree 8.
\end{thm}

\begin{proof} % [Theorem \ref{tight_3}] \quad
Let $\A$ be the representation matrix in standard form 
of an $\SHF(8;8,2,\{1,3\})$.
Then there are
$T={8 \choose 3}5=280$ column pairs to be separated. 
On average,
each row separates $\frac{280}{8}=35$ column pairs.
From Lemma \ref{separate.lem}, a type 1 row separates
$35$ column pairs, a type 2 row separates $40$ column pairs,
a type 3 row separates $35$ column pairs, and a type 4 row 
separates $32$ column pairs.

Suppose that $\A$ contains a row of type 1. Then by 
Theorem \ref{tight_1} and Lemma \ref{type_1}, $\A$ is a permutation matrix.
Therefore we can  assume that $\A$ contains no row of type 1.
We next show that $\A$ contains a row of type 2. 
Assume the contrary and suppose without loss of generality that $\mu_1 \geq \max \{\mu_2, \dots ,\mu_8\}$.
Since, on average,
each row of $\A$ separates 35 column pairs, it follows that all rows must be of 
type 3. However, from Lemma \ref{overlapping}, it can be verified that any two rows of type
3 must separate a positive number of  common column pairs, so we have a contradiction.

Therefore, by permuting columns if necessary, we may assume that the first
row of $\A$ is of type 2, having the entry 1 in the first two columns. Thus we have
\begin{displaymath}
\A =
\left( \begin{array}{cc|ccccccc}
    1 & 1   & 0 & 0 & 0 & 0 & 0 & 0  \\ \hline
    a_2 & b_2   &    &    &    &    &   & \\
    a_3 & b_3   &    &    &    &    &   & \\
    a_4 & b_4   &    &    &    &    &   & \\
    a_5 & b_5   &    &    &   \B    &   & \\
    a_6 & b_6   &    &    &    &    &   & \\
    a_7 & b_7   &    &    &    &    &   & \\
    *   &  *    &    &    &    &    &   &      
    \end{array} \right)
\end{displaymath}
where $\B$ is the representation matrix of an $\SHF(7;6,2,\{1,3\})$.
By Lemma \ref{7r6c} we may assume
\begin{displaymath}
\B =
\left( \begin{array}{cccccc}
    1 & 0 & 0 & 0 & 0 & 0 \\
    0 & 1 & 0 & 0 & 0 & 0 \\
    0 & 0 & 1 & 0 & 0 & 0 \\
    0 & 0 & 0 & 1 & 0 & 0 \\
    0 & 0 & 0 & 0 & 1 & 0 \\
    0 & 0 & 0 & 0 & 0 & 1 \\ 
    * & * & * & * & * & *
    \end{array} \right).
\end{displaymath}
From the assumption that $\A$ has no row of type 1, we have that
$(a_i,b_i)\not= (0,0)$ for  $i=2,\ldots, 7$. It follows that rows
$2,3, \ldots, 7$ are of type 2 or 3 and each of them overlap row 1.
 From Lemma \ref{overlapping}, if $(a_i,b_i)=(1,1)$, then row $i$ separates at most $15$
new column pairs, whereas if $(a_i,b_i)=(1,0) \mbox{ or } (0,1)$,
then row $i$ separates at most $30$ new column pairs.
In any case, rows $1,2, \ldots, 7$ separate at most
$40+ 6 \times 30=220$ column pairs. Thus row 8 has to separate
at least $280-220=60$ new column pairs, which is impossible.
This completes the proof.  
\end{proof}

The next theorem follows immediately from Theorem \ref{tight_3} and 
Theorem \ref{perm_bound}.

\begin{thm} \label{bound_3}
Suppose there exists an $\SHF(8;n,2,\{1,3\})$. Then $n \leq 8.$
\end{thm}

\subsection{The case  $w=3$ and $N=9$ }

We first prove several preliminary lemmas.
\begin{lem} \label{7r5c}
Let $\A$ be the representation matrix in standard form 
of an $\SHF(7; 5, 2, \{1,3\})$
Then by permuting the rows of $\A$ we have that the first 5 rows
are of type 1 and the last two row are of any type.
\end{lem}

\begin{proof} There are $2 \binom{5}{3} = 20$ column pairs to be separated.
A type 1 row separates four column pairs
 and a type 2 row separates two column pairs. %As $\A$ has 7 rows,
 %on average each row separates $\alpha= \frac{20}{7} > 2$
 %column pairs. 
 Let $x$ denote the number of disjoint type 1 rows;
 we claim that $x = 5$.
 If $x \leq 2$, then the number of column pairs that are separated
 is at most $2 \times 4 + 5 \times 2 = 18 < 20$, which is a contradiction.
 
 Suppose $x=4$. Four disjoint type 1 rows separate 16 column pairs.
 Due to overlap, any type two row  separates at most one additional 
 column pair. This means that at least four type two rows are required
 so that all column pairs are separated. This yields eight rows, which is
 a contradiction.
 
 Finally, we suppose $x=3$. Three disjoint type 1 rows separate 12 column pairs.
 There is one possible type 2 row that separates two additional column pairs,
 and any other type two row  separates at most one additional 
 column pair. It follows that we cannot separate all the column pairs using 
 seven rows. 
\end{proof}

\begin{lem} \label{8r6c}
Let $\A$ be the representation matrix in standard form 
of an $\SHF(8; 6, 2, \{1,3\})$
Then by permuting the rows of $\A$ we have that the first six rows
are of type 1 and the last two rows are of any type.
\end{lem}

\begin{proof}
There are in total $T={6 \choose 3}3=60$ column pairs 
of $\A$ to be separated. %On average each row separates
%$\alpha = \frac{60}{8} > 7$ pairs; 
A type 1 row
separates ${5\choose 3}=10$ column pairs. A type 2 row 
separates $2{4\choose 3}=8$ columns pairs, but it separates
at most seven new column pairs if it is not disjoint from all the
1 and type 2 rows (Lemma \ref{overlapping}). 
Finally, a type 3 row
separates $3{3\choose 3}\times 2=6$ pairs,
but it separates at most five new column pairs if it is not disjoint from all the
type 1 rows (Lemma \ref{overlapping}).

First, suppose there is no row of type 1. In order to cover 
all the column pairs, we need at least six rows of type 2
(observe that $5 \times 8 + 3 \times 6 = 58 < 60$).
There are at most three disjoint rows of type two. Therefore
there are at least three rows of type two that each cover 
at most four new pairs. As well, there are two additional
rows that each cover at most six new column pairs.
The maximum number of column pairs that are covered is
$3 \times 8 + 3 \times 7 + 2 \times 6 = 57 < 60$, so 
we have a contradiction.

Thus we may assume that the first row of
$\A$ is of type 1, with entry 1 in the first column, so 
$\A$ has the form
\begin{displaymath}
\A =
\left( \begin{array}{c|ccccc}
    1 & 0 & 0 & 0 & 0 & 0  \\ \hline
    a &    &    &    &    &   \\
    b &     &    &    &    &   \\
    c &    &    &    &    &   \\
    d &    &    &   \B    &   \\
    e &    &    &    &    &   \\
    f &     &    &    &    &   \\
    g &    &    &    &    &    
    \end{array} \right)
\end{displaymath}
where $\B$ is the representation matrix of an $\SHF(7;5,2,\{1,3\})$.
By Lemma \ref{7r5c} we may assume
\begin{displaymath}
\B =
\left( \begin{array}{ccccc}
    1 & 0 & 0 & 0 & 0 \\
    0 & 1 & 0 & 0 & 0 \\
    0 & 0 & 1 & 0 & 0 \\
    0 & 0 & 0 & 1 & 0 \\
    0 & 0 & 0 & 0 & 1 \\
    * & * & * & * & * \\ 
    * & * & * & * & * 
    \end{array} \right)
\end{displaymath}
Let $x$ denote the number of $1$'s in the multiset $\{a,b,c,d,e\}$.
We want to show that $x=0$.

First suppose $x=1$ and  
assume without loss of generality that $a=1$. There are six column pairs not covered
by the first six rows, namely, $(\{2\}, \{1,y,z\})$, where
$\{y,z\} \subseteq \{3,4,5,6\}$. The only way that these six column pairs 
can be covered by two rows is if one of the two rows is of type 
$1$, having a $1$ in column $2$. Thus we have six rows of type 1 and  this case is done.

Next, suppose $x=2$ and  
assume without loss of generality that $a=b=1$. There are twelve column pairs not covered
by the first six rows, namely, $(\{2\}, \{1,y,z\})$, where
$\{y,z\} \subseteq \{3,4,5,6\}$
and $(\{3\}, \{1,y,z\})$, where
$\{y,z\} \subseteq \{2,4,5,6\}$. The only way to cover the first six column pairs 
by two rows is to include the row of type 
$1$, having a $1$ in column $2$. Further, the only way to cover the second six column pairs 
by two rows is to include the row of type 
$1$, having a $1$ in column $3$. Thus we have six rows of type 1 and  this case is done.

If $x \geq 3$, then we need $x$ additional rows of type 1 to cover the uncovered 
column pairs, but now the total number of rows is $6 + x > 8$. So these cases cannot occur,
and the proof is complete.
\end{proof}

\begin{lem} \label{8r7c}
Let $\A$ be the representation matrix of an $\SHF(8;7,2,\{1,3\})$
in standard form. By permuting the rows of $\A$ if necessary 
we have that the
first seven rows of $\A$ are of type 1. The last row
can be of any type.
\end{lem}

\begin{proof}
There are $T={7 \choose 3}4 =140$ column pairs of $\A$ to be separated. 
%On average, each row separates
%$\alpha= T/8 > 17$ column pairs. 
A type 1 row separates
$20$ column pairs, a type 2 row separates 
$20$ column pairs, and a type 3 row separates 
$16$ column pairs. Further, if a type 2 row is not disjoint from all other
type two rows, then it separates at most 16 new column pairs (Lemma \ref{overlapping}).

First we show that there must be a row of type 1. 
Suppose not; then there are $x$ rows of type 2 and
$8-x$ rows of type $3$. Since $2 \times 20 + 6 \times 16 = 136 < 140$,
we must have $x \geq 3$. Now, there can be at most three disjoint
rows of weight 2, so the  number of column pairs
covered is at most $3 \times 20 + (x-3)16 + (8-x)16 = 140$.
Then, in order for all 140 column pairs to be separated, each row of type
3 must be disjoint from all rows of type 2 (Lemma \ref{overlapping}), which is impossible.

Therefore,  we may assume that the first row of $\A$ is of type 1 with entry 1
in the first column. Thus $\A$ has the form
\begin{displaymath}
\A =
\left( \begin{array}{c|cccccc}
    1 &    0 & 0 & 0 & 0 & 0 & 0  \\ \hline
    a_2 &     &    &    &    &   & \\
    a_3 &     &    &    &    &   & \\
    a_4 &     &    &    &    &   & \\
    a_5 &     &    &   \B    &   & \\
    a_6 &     &    &    &    &   & \\
    a_7 &     &    &    &    &   & \\
    b   &     &    &    &    &   &      
    \end{array} \right)
\end{displaymath}
where $\B$ is the representation matrix of an $\SHF(7;6,2,\{1,3\})$.
By Lemma \ref{7r6c} we have that
\begin{displaymath}
\B =
\left( \begin{array}{cccccc}
    1 & 0 & 0 & 0 & 0 & 0 \\
    0 & 1 & 0 & 0 & 0 & 0 \\
    0 & 0 & 1 & 0 & 0 & 0 \\
    0 & 0 & 0 & 1 & 0 & 0 \\
    0 & 0 & 0 & 0 & 1 & 0 \\
    0 & 0 & 0 & 0 & 0 & 1 \\ 
    * & * & * & * & * & *
    \end{array} \right)
\end{displaymath}
 To prove the lemma we show that $a_2 = \cdots = a_7 = 0$.
 
Suppose some $a_i$ is nonzero, say $a_2 = 1$.
The column pairs not covered
by the first seven rows, include the following ten column pairs: $(\{2\}, \{1,y,z\})$, where
$\{y,z\} \subseteq \{3,4,5,6,7\}$. The only way that these ten column pairs 
can be covered by two rows is if one of the two rows is of type 
$1$, having a $1$ in column $2$. Thus we must have a row of type 1 
whose nonzero entry is in column $2$. The above argument can be applied
for any $a_i = 1$, which completes the proof.
\end{proof}

We are now in a position to prove the following theorem.
\begin{thm}\label{tight_4}
The representation matrix $\A$ in standard form of an 
$\SHF(9;9,2,\{1,3\})$ is a permutation matrix.
\end{thm}

\begin{proof}
 %We first give  some counting facts.
 There are in total $6 \binom{9}{3} = 504$ column pairs 
of $\A$ to be separated. %On average each row separates
%$\alpha = \frac{504}{9} > 55$ column pairs; 
A type 1 row
separates $56$ column pairs, a type 2 row
separates $70$ column pairs, a type 3 row
separates $66$ column pairs, and a type 4 row
separates $60$ column pairs. If a type 2 row overlaps another type two row,
then it separates at most 50 new column pairs, and if a 
type 3 row has overlap 2 with a type 2 row, then it separates 
at most 26 new column pairs (Lemma \ref{overlapping}).

 If $\A$ has a row of type 1, then by Lemma \ref{type_1}, $\A$
 is a permutation matrix and we are done.
 We will show that $\A$ must contain a row of type 1 by successively ruling out
 the cases that $\A$ contains a row of type 2, type 3 or type 4.

\begin{description}
\item[Case 1]  $\A$ contains row of type 2 but no rows of type 1.\mbox{\quad}
  
Assume w.l.o.g.\ that the first row of $\A$ is of type 2 with  
entry 1 in columns 1 and 2. By removing the first two columns 
and the first row of $\A$ we obtain an $8\times 7$ binary matrix
$\B$ which is the representation matrix of an $\SHF(8;7,2,\{1,3\})$.
By Lemma \ref{8r7c}, we may assume that the first seven rows of
$\B$ are of type 1. Here is the structure of the first eight rows of $\A$:
\begin{displaymath}
\left( \begin{array}{cc|ccccccc}
     1 & 1 &  0 & 0 & 0 & 0 & 0 & 0 & 0 \\ \hline
     * & * &  1 & 0 & 0 & 0 & 0 & 0 & 0 \\
     * & * &  0 & 1 & 0 & 0 & 0 & 0 & 0 \\
     * & * &  0 & 0 & 1 & 0 & 0 & 0 & 0 \\
     * & * &  0 & 0 & 0 & 1 & 0 & 0 & 0 \\    
     * & * &  0 & 0 & 0 & 0 & 1 & 0 & 0 \\
     * & * &  0 & 0 & 0 & 0 & 0 & 1 & 0 \\
     * & * &  0 & 0 & 0 & 0 & 0 & 0 & 1   
    \end{array} 
    \right)
\end{displaymath}

First note that rows $2, \ldots, 8$ must contain an entry equal to 1 in the first
two columns, since we are assuming that $\A$ has no rows of type 1.
%This shows that $(*,*) \in \{(0,1), (1,0), (1,1) \}$. 
This 
implies that rows $2, \ldots, 8$ are all of type 2 or 3 and they all
overlap row 1. 
If a type 2 row overlaps another type two row,
then it separates at most 50 new column pairs, and if a 
type 3 row has overlap 2 with a type 2 row, then it separates 
at most 26 new column pairs (Lemma \ref{overlapping}).
Therefore,  the first
eight rows of $\A$ can separate at most $70+ 50\times7=420$ column pairs.
Then the last row of $\A$ has to separate at least
$504-420=84$ column pairs, which is impossible. This rules out Case 1.      

\item[Case 2]  $\A$ contains row of type 3 but no rows of type 1 or 2.\mbox{\quad}
  
Assume that the first row of $\A$ is of type 3 with  
entry 1 in columns 1 and 2 and 3. By removing the first three columns 
and the first row of $\A$ we obtain an $8\times 6$ binary matrix
$\B$ which is the representation matrix of an $\SHF(8;6,2,\{1,3\})$.
By Lemma \ref{8r6c}, we may assume that the first six rows of
$\B$ are of type 1. Here is the structure of the first seven rows of $\A$:
\begin{displaymath}
\left( \begin{array}{ccc|cccccc}
     1 & 1 & 1  & 0 & 0 & 0 & 0 & 0 & 0  \\ \hline
     * & * & *  & 1 & 0 & 0 & 0 & 0 & 0  \\
     * & * & *  & 0 & 1 & 0 & 0 & 0 & 0  \\
     * & * & *  & 0 & 0 & 1 & 0 & 0 & 0  \\
     * & * & *  & 0 & 0 & 0 & 1 & 0 & 0  \\    
     * & * & *  & 0 & 0 & 0 & 0 & 1 & 0  \\
     * & * & *  & 0 & 0 & 0 & 0 & 0 & 1  \\   
    \end{array} \right)
\end{displaymath} 
Note that $\A$ only has rows of type 3 or 4. 
Let $i \in \{2,3,4,5,6,7\}$. 
If row $i$ is of type 3, then it has overlap 2 with row 1 and it 
separates at most $66-20=46$ new column pairs, and 
if row $i$ is of type 4, then it has overlap 3 with row 1 and separates at most $66-30=36$ column pairs
(Lemma \ref{overlapping}).
It follows that the first seven rows of $\A$ separate at most
$66+6\times 46 =342$ column pairs. Hence rows 8 and 9 have to separate
at least $504-342 = 162$ new column pairs, which is impossible. This rules out
Case 2.

\item[Case 3]  All rows of $\A$ are of type 4.\mbox{\quad}

Let $\alpha_i$ denote the number of repeated column pairs arising from 
two rows of type 4 having overlap equal to $i$. From Lemma \ref{overlapping}, we have 
$\alpha_0 = 32$, $\alpha_1 = 6$, $\alpha_2 = 2$, and $\alpha_3 = 16$.
So the maximum number of column pairs covered by any two rows is
$60 + 58 = 118$.

Let's now consider sets of three rows. A consideration of possible cases shows
that the maximum number of column pairs covered by three rows is
$60 + 58 + 56 = 174$. This happens if and only if the three rows have
pairwise overlaps all equal to $2$. There are in fact three non-isomorphic ways
in which this can happen:
\[
\begin{array}{ccc}
\begin{array}{c}
111100000\\
110011000\\
110000110
\end{array} & 
\begin{array}{c}
111100000\\
110011000\\
101010100
\end{array} & 
\begin{array}{c}
111100000\\
110011000\\
001111000
\end{array}
\end{array}
\]
The three cases are distinguished by the number of columns of weight 3.
 
Now let's look at the maximum number of column pairs obtained by extending
one of the three 3-row configurations enumerated above.
The maximum number of column pairs covered by four such rows is
$60 + 58 + 56 + 54 = 228$. This happens if and only if the four rows have
pairwise overlaps all equal to $2$. There are in fact four non-isomorphic ways
in which this can happen:
\[
\begin{array}{cccc}
\begin{array}{c}
111100000\\
110011000\\
110000110\\
101010100
\end{array} & 
\begin{array}{c}
111100000\\
110011000\\
101010100\\
100101100
\end{array} & 
\begin{array}{c}
111100000\\
110011000\\
101101000\\
011011000
\end{array} &
\begin{array}{c}
111000100\\
110110000\\
101101000\\
011011000
\end{array}
\end{array}
\]
%%%%%%%%%%%%%%%%%%%%%%%%%
 Now suppose the rows are ordered so $\mu_1 \geq \mu_2 \geq \dots \geq \mu_9$.
 We know that $\mu_1 = 60$ and $\mu_2 \leq 58$.
 We consider three cases and apply the results above.
 \begin{enumerate}
 \item If $\mu_1 = 60$, $\mu_2 = 58$ and $\mu_3 = 56$,  then $\mu_4 \leq 54$.
 Then \[\sum \mu_i \leq 60 + 58 + 56 + 6 \times 54 = 498 < 504.\] So this case is
 impossible.
 \item If $\mu_1 = 60$, $\mu_2 = 58$ and $\mu_3 \leq 55$, then 
 \[ \sum \mu_i \leq 60 + 58 + 7 \times 55 = 503 < 504.\] So this case is
also impossible.
\item If $\mu_1 = 60$ and $\mu_2 < 58$, then  $\mu_2 \leq 54$ and  
 \[\sum \mu_i \leq 60 + 8 \times 54 = 492 < 504.\]
 \end{enumerate}
 Since all cases lead to a contradiction, the proof is complete.
\end{description}
\end{proof}
 
\begin{thm}\label{bound_4}
Suppose there is an $\SHF(9;n,2,\{1,3\})$. Then $ n \leq 9$.
\end{thm}

\begin{proof}
 Theorem \ref{bound_4} follows from Theorems \ref{perm_bound} and 
  \ref{tight_4}.    
\end{proof}

\section{ Discussion of the case $w=2$ }
\label{w=2.sec}

For completeness, we include a discussion regarding the $w=2$ case.
Since $q = w$, some of the previously known results apply, and the
situation is much different from where $w \geq 3$.

\begin{thm} \label{Nplus1_shf}
For every $N \geq 3$, there exists an $\SHF(N;N+1,2,\{1,2\})$.
\end{thm}

\begin{proof}
Take the $N \times N$ identity matrix and append to it a column of 1s;
call this matrix $\A$. We will show that $\A$ is an $\SHF(N;N+1,2,\{1,2\})$.

Let $(C_1=\{x\},C_2=\{y,z\})$ be a column set pair. First consider
$1 \leq x \leq N$. If $1 \leq y,z \leq N$ then $(C_1,C_2)$ is clearly 
separated by $\A$. Suppose w.l.o.g. that $z = N+1$, then row $y$ has entry
1 in columns $y,z$ and entry 0 in column $x$, so $(C_1,C_2)$ is again
separated.

Finally, consider $x = N+1$, so $1 \leq y,z \leq N$. Since $N \geq 3$,
there is some row $w \notin \{y,z\}$, so row $w$ has entry 0 in columns
$y,z$ and entry 1 in column $x$, so $(C_1,C_2)$ is separated. 
\end{proof}

%\vspace{4mm}
Theorem \ref{Nplus1_shf} above shows that Theorem \ref{bound_1} does not
hold when $w=2$. We will also demonstrate that Theorem \ref{tight_1} and
Theorem \ref{perm_bound} do not hold when $w=2$.

\begin{thm} \label{non_perm_shf}
The matrix
\begin{displaymath}
\A =
\left( \begin{array}{cccc}
     1 & 1 & 0 & 0 \\
     0 & 1 & 1 & 0 \\ 
     1 & 0 & 1 & 0 \\
     0 & 0 & 0 & 1 \\  
    \end{array} \right)
\end{displaymath}
is an $\SHF(4;4,2,\{1,2\})$.
\end{thm}

The result in Theorem \ref{non_perm_shf} can be extended to $N > 4$ by constructing
the matrix
\begin{displaymath}
\B =
\left( \begin{array}{cc}
     \A & 0 \\
     0 & I_k \\  
    \end{array} \right)
\end{displaymath}
where $\A$ is from Theorem \ref{non_perm_shf} and $I_k$ is the $k \times k$ identity
matrix for $k = N-4$. Observe that for every column $x,y \in \{1,2,3\}$, there exist
rows $r_x,r_y$ such that $r_x(x) = 1$, $r_x(y) = 0$ and $r_y(x) = 0$, $r_y(y) = 1$.
It is straightforward to verify that $\B$ is indeed an $\SHF(N;N,2,\{1,2\})$.
Theorem \ref{no_perm_bound} below covers the last
case $N = 3$, and shows that Theorem \ref{perm_bound} does hold when $w=2$.

\begin{thm} \label{no_perm_bound}
The representation matrix of an $\SHF(3;3,2,\{1,2\})$ in standard form is
a permutation matrix.
\end{thm}

\begin{proof}
In standard form, every row is of type 1. Two distinct rows must not overlap,
so each column also has one 1. 
\end{proof}

\section{Conclusion}
\label{conclusion.sec}

Gathering together the results proven in this paper, we have the following theorems.

\begin{thm}
Let $w$, $N$ be positive integers such that $w \geq 3$
and $w+1 \leq N \leq 3w$. Suppose there exists an
$\SHF(N; n, 2, \{1,w\})$. Then  $n \leq N.$
\end{thm}

\begin{thm}
Let $w$, $N$ be positive integers such that $ w \geq 3$
and $w+1 \leq N \leq 3w$. Suppose there exists an
$\SHF(N; n, 2, \{1,w\})$ with $n=N$. Then its representation
matrix in standard form is a permutation matrix of degree $N$.
\end{thm}

Here is an interesting problem that is suggested by our work: 
For a given $w$, find the smallest $N$ such that there exists an
$\SHF(N; n,2, \{1,w \})$  with $n > N$. A closely related problem is
to find the smallest $n$ such that there exists an
$\SHF(n; n,2, \{1,w \})$  that is not a permutation matrix.
Finally, it may be of interest to try to generalize the results in this
paper to SHF of other types, or to SHF over non-binary alphabets.

%\vspace{2mm}\noindent

\end{document}